\documentclass{ws-jnmp}

\begin{document}

%\arttype{Letter} % default 'Article'

\markboth{V. K. Chandrasekar, A. Durga Devi and M. Lakshmanan}
{Recursive Generation of Isochronous Hamiltonian Systems}

\catchline{1}{1}{2009}{}{}
%
%%%%%%%%%%%%%%%%%%%%%%%%%%%%%%%%%%%%%%%%%%%%%%%%%%%%%%%%%%%%%%%%%%%%
\copyrightauthor{M. LAKSHMANAN}

\title{RECURSIVE  GENERATION OF ISOCHRONOUS HAMILTONIAN SYSTEMS}

\author{\bf V. K. CHANDRASEKAR$^{\dag}$, A. DURGA DEVI$^*$,  M. LAKSHMANAN$^{\ddag}$}

\address{Centre for Nonlinear Dynamics,\\ School of Physics,\\
Bharathidasan University,\\ Tiruchirappalli - 620 024, India.\\
\email{$^{\dag}$sekar@cnld.bdu.ac.in\\$^*$durga@cnld.bdu.ac.in\\$^{\ddag}$lakshman@cnld.bdu.ac.in}}

%\author{\bf A. DURGA DEVI}

%\address{Centre for Nonlinear Dynamics,\\ School of Physics,\\
%Bharathidasan University,\\ Tiruchirappalli - 620 024, India.\\
%\email{durga@cnld.bdu.ac.in}}

%\author{\bf M. LAKSHMANAN}

%\address{Centre for Nonlinear Dynamics,\\ School of Physics,\\
%Bharathidasan University,\\ Tiruchirappalli - 620 024, India.\\
%\email{lakshman@cnld.bdu.ac.in}}

\maketitle

\begin{abstract}
We propose a simple procedure to identify the collective coordinate $Q$ which is used to generate the isochronous Hamiltonian. The new isochronous Hamiltonian generates more and more isochronous oscillators, recursively.
\end{abstract}

%\keywords{Keyword1; keyword2; keyword3.}

%\ccode{2000 Mathematics Subject Classification: 22E46, 53C35, 57S20}

In recent years considerable interest has been shown to identify and classify isochronous systems. In this direction Calogero and his coworkers have introduced a number of systematic procedures to generate isochronous oscillator systems\cite{Calogero:08,Calogero:08f,Calogero:08g,Calogero:08d,Calogero:08c,Calogero:07,Calogero:08a}.  In a different direction the existence of amplitude independent frequency of nonlinear oscillators have been identified using nonlocal transformations\cite{vkc:09,vkc:09a}.  Recently
 Calogero and Leyvraz \cite{Calogero:07,Calogero:08a} proposed a new powerful technique to generate isochronous Hamiltonian systems. In this technique they have shown that the real autonomous Hamiltonian $H(p,q)$ can be transformed to an $\Omega$-modified Hamiltonian, that is, $H_{(1)}=\frac{1}{2}(H(p,q)^2+\Omega^2 Q(p,q)^2)$, which has the isochronous property. Here $H$ behaves as the new momentum and $Q$ is the canonically conjugate/collective coordinate conjugate to the Hamiltonian, such that the Poisson bracket \{H,Q\}=1. $\Omega$ is an arbitrary constant. Due to the nature of the $\Omega$-modified Hamiltonian system, now the new momentum $H$ and coordinate $Q$ evolve periodically with period $T=2\pi/\Omega$, and so the momentum $p$ and coordinate $q$ also evolve periodically with the same period. In \cite{Guha:09}, the authors have shown the interesting connection between symplectic rectification and isochronous Hamiltonian systems.

In this brief communication, we propose a simple procedure to identify the collective coordinate $Q$ which is used to generate the isochronous Hamiltonian. We also point out the further interesting possibility to  generate recursively more isochronous oscillators from the newly constructed  isochronous Hamiltonian. We also illustrate this possibility with an example. To identify $Q$ we start with the following theorem:

\newtheorem{thm}{Theorem}
\begin{thm} If $H=H(p,q)$ is the Hamiltonian of a given system such that it can be inverted to find a single valued $q$ or $p$ in terms of the other variable and $H$ explicitly, then the system admits at least one integral of the form $I=-t+Q(p,q)$, and $Q(p,q)$ is a collective coordinate conjugate to the Hamiltonian in an appropriate phase space (avoiding multivaluedness and singularities).
\end{thm}

\begin{proof}
The Hamilton equations of the Hamiltonian $H(p,q)$ are
\begin{eqnarray}
&&\dot{q}=\frac{\partial H}{\partial p}=f_1(p,q),\quad
\dot{p}=-\frac{\partial H}{\partial q}=f_2(p,q).\label {lam101}
\end{eqnarray}
Now inverting the Hamiltonian $H(p,q)$ in terms of $p$ or $q$ and substituting the resultant expression into the right hand side of the $\dot{q}$ or $\dot{p}$ equation we get
\begin{eqnarray}
&&\dot{q}=f_1(p,q)=f_3(q,H),\quad \mbox{or} \quad \dot{p}=f_2(p,q)=f_4(p,H).\label {lam102}
\end{eqnarray}
Integrating the above equation we get
\begin{eqnarray}
&&I+t=\int \frac{dq}{f_3(q,H)}=Q_{1}(p,q),\quad \mbox{or} \quad 
I+t=\int \frac{dp}{f_4(p,H)}=Q_{2}(p,q)\label {lam103}
\end{eqnarray}
where $I$ is the integration constant and $Q_{1}$ and $Q_{2}$ are two of the possible collective coordinates (where the phase space chosen such that the coordinates are single valued and non-singular).  The latter fact can be easily proved by noting that the total differentiation of any one of the $Q(p,q)$ yields on using (\ref{lam101}),
\begin{eqnarray}
\frac{d Q}{d t}=\dot{q}\frac{\partial Q}{\partial q}+\dot{p}\frac{\partial Q}{\partial p}=\{H,Q\}=1.\label {lam104}
\end{eqnarray}
Thus $Q_{1}$ or $Q_{2}$ is canonically conjugate to $H$ and can serve as the required collective coordinate.
\end{proof}

Now considering the above Hamiltonian $H$ as a new momentum and $Q$ as a new collective coordinate and substituting these into the $\Omega$-modified isochronous Hamiltonian $H_{(1)}$ given by Calogero and Leyvraz \cite{Calogero:07,Calogero:08a},
\begin{eqnarray}
&&H_{(1)}=\frac{1}{2}(H(p,q)^2+\Omega^2 Q(p,q)^2),\label {lam105}
\end{eqnarray}
one can show that the dynamical system in $p$ and $q$ is indeed isochronous. This has been proved in \cite{Calogero:07,Calogero:08a}. 

\subsubsection*{Example:1}

Let us consider $H=pq$, then $Q=\log q$. Then the $\Omega$-modified Hamiltonian $H_{(1)}=\frac{1}{2}(p^2q^2+\Omega^2 (\log q)^2)$  is isochronous.

Next we note that the collective coordinate $Q_{(1)}=\frac{1}{\Omega}\tan^{-1}\bigg(\frac {\Omega Q}{H}\bigg)$, confined to the principal branch of the right hand side, is conjugate to the Hamiltonian $H_{(1)}$ which is obtained using Theorem 1, and Eq.(0.3).  Then we note the following theorem.

\begin{thm} Let $H_{(1)}-a_{(1)}$ and $Q_{(1)}=\frac{1}{\Omega}\tan^{-1}\bigg(\frac {\Omega Q}{H}\bigg)$, where the latter is confined to the principal branch of the arctan function, be the new momentum and the collective coordinates, respectively,  then the $\Omega_{(1)}$- modified Hamiltonian $H_{(2)}=\frac{1}{2}[(H_{(1)}-a_{(1)})^2+\Omega_{(1)}^2 Q_{(1)}^2]$ also has isochronous dynamics, where $a_{(1)}$ and  $\Omega_{(1)}$ are suitable arbitrary positive system parameters. 
\end{thm}

\begin{proof}
From the nature of $H_{(2)}$, the solutions for $H_{(1)}-a_{(1)}$ and $Q_{(1)}$ are written as
\begin{eqnarray}
&&H_{(1)}-a_{(1)}=A_{(1)}\cos[\Omega_{(1)} t+\delta_{(1)}],\qquad Q_{(1)}=\frac{A_{(1)}}{\Omega_{(1)}}\sin[\Omega_{(1)} t+\delta_{(1)}],\label {lam106}
\end{eqnarray}
where $A_{(1)}$ and $\delta_{(1)}$ are arbitrary constants.
Substituting equation (\ref{lam106}) into the expressions for $H_{(1)}$ (equation (\ref{lam105})) and $Q_{(1)}$ and inverting we get the solutions for $H$ and $Q$ in the forms
\begin{eqnarray}
&&H=\sqrt{2a_{(1)}+2A_{(1)}\cos(\Omega_{(1)} t+\delta_{(1)})}
\cos[\frac{\Omega A_{(1)}}{\Omega_{(1)}}\sin(\Omega_{(1)} t+\delta_{(1)})]\nonumber\\
&&Q=\frac{1}{\Omega}\sqrt{2a_{(1)}+2A_{(1)}\cos(\Omega_{(1)} t+\delta_{(1)})}\sin[\frac{\Omega A_{(1)}}{\Omega_{(1)}}\sin(\Omega_{(1)} t+\delta_{(1)})].\label {lam107}
\end{eqnarray}
These solutions evolve periodically with period $T=2\pi/\Omega_{(1)}$, for $a_{(1)}>|A_{(1)}|$ so that the quantity inside the square root remain positive for all times. The expressions for $p$ and $q$ can be obtained upon inverting $H$ and $Q$. Now, as we already know that $H$ and $Q$ evolve periodically, it is obvious that $p$ and $q$ must also evolve periodically with the same period, namely $T_{(1)}=2\pi/\Omega_{(1)}$, but in general different from $T_{(0)}=2\pi/\Omega$.
\end{proof} 

From the above Theorem 2, one can identify recursively the $\Omega_{(i)}$ modified Hamiltonian from $H_{(i+1)}=\frac{1}{2}[(H_{(i)}-a_{(i)})^2+\Omega_{(i)}^2 Q_{(i)}^2]$, $i=0,1,2......n$, where $H_{(0)}=H$, $Q_{(0)}=Q$, $\Omega_{(0)}=\Omega$, and  $a_{(0)}=0$.  Here $H_{(i)}-a_{(i)}$ and $Q_{(i)}$ are the new momentum and its corresponding canonically 
conjugate/collective coordinate, respectively. All the above systems yield periodic solutions with period $T_{(i)}=2\pi/\Omega_{(i)}$ and they can be deduced using the relations,

\begin{eqnarray}
H_{(i)}=a_{(i)}+\sqrt{2 H_{(i+1)}}\cos[\Omega_{(i)}Q_{(i+1)}],\; Q_{(i)}=\frac{1}{\Omega_{(i)}}(\sqrt{2 H_{(i+1)}}\sin[\Omega_{(i)}Q_{(i+1)})],\;i=0,1,...n\nonumber
\end{eqnarray}
 and also
\begin{eqnarray}
 H_{(n)}=a_{(n)}+A_{(n)}\cos[\Omega_{(n)} t+\delta_{(n)}],\qquad Q_{(n)}=\frac{A_{(n)}}{\Omega_{(n)}}\sin[\Omega_{(n)} t+\delta_{(n)}]. 
\end{eqnarray}
Here $A_{(n)}$ and $\delta_{(n)}$ are arbitrary constants, $\Omega_{(i)}$ and $a_{(i)}$, $i=0,1,2......n$, are \\
system (arbitrary) parameters. Note that for the solution to remain real, one has to impose the condition $a_{i-1}>\sqrt{2(a_{i}+X_{i})}=X_{i-1}$ and $a_{(n)}>|A_{n}|=X_n$, $i=1,2, ... n$. Using the above periodic solutions, one can easily see that the canonical variables $p$ and $q$ also evolve periodically with period $T_{n}=2\pi/\Omega_{(n)}$.

We now illustrate the above recursive procedure with an example.

\subsubsection*{Example:2} Let us consider the Hamiltonian $H=p^{n}g(q)$, where $g(q)$ is an arbitrary function of $q$, for which the Hamilton's equations can be written as
\begin{eqnarray}
&&\dot{q}=\frac{\partial H}{\partial p}=np^{n-1}g(q),\quad \dot{p}=-\frac{\partial H}{\partial q}=-p^ng'(q).\label {ex101}
\end{eqnarray}
Here $g'(q)=\frac{dg}{dq}$.
Note that the integration of the equation $dp/dq=-pg'(q)/(ng(q))$ (vide equation (\ref{ex101})) gives the integration constant $I=p^{n}g(q)$ which is nothing but the Hamiltonian $H$.
From the Hamiltonian we get $p=(H/g(q))^{\frac{1}{n}}$ and substituting this expression into the $\dot{q}$ equation, we obtain
\begin{eqnarray}
&&\dot{q}=n(H/g(q))^{\frac{n-1}{n}}g(q)=n(H)^{\frac{n-1}{n}}g(q)^{\frac{1}{n}}.\label {ex102}
\end{eqnarray}
Integrating (\ref{ex102}) we get
\begin{eqnarray}
&&I+t=\frac{g(q)^{\frac{1-n}{n}}}{np^{n-1}}\int g(q)^{-\frac{1}{n}}dq.\label {ex103}
\end{eqnarray}
Now the collective coordinate $Q(p,q)$ is of the form
\begin{eqnarray}
&&Q(p,q)=\frac{g(q)^{\frac{1-n}{n}}}{np^{n-1}}\int g(q)^{-\frac{1}{n}}dq,\label {ex104}
\end{eqnarray}
which is conjugate to the Hamiltonian, that is $\{H,Q\}=1$. This is in conformity with Theorem 1.

For simplicity let us consider the case $n=1$ and $g(q)=q$. In this case the new momentum and the collective coordinate are written as $H=pq$ and $Q=\log(q)$. Substituting these into the $\Omega$ - modified Hamiltonian $H_{(1)}$ given in (\ref{lam105}) we get 
\begin{eqnarray}
&&H_{(1)}=\frac{1}{2}((pq)^2+\Omega^2 \log(q)^2).\label {lam105a}
\end{eqnarray}
Using the procedure given in \cite{Calogero:07,Calogero:08a} or following our procedure given above the solution for $p$ and $q$ now become
\begin{eqnarray}
&&p(t)=A\cos(\Omega t+\delta)e^{-\frac{A}{\Omega}\sin(\Omega t+\delta)},\qquad q(t)=e^{\frac{A}{\Omega}\sin(\Omega t+\delta)}\label {lam105b}
\end{eqnarray} 
which are periodic with period $T=2\pi/\Omega$, so the system for $p$ and $q$ is isochronous.

Now consider the  $\Omega$ - modified Hamiltonian $H_{(1)}-a_{(1)}$ as the new momentum and $Q_{(1)}=\frac{1}{\Omega}\tan^{-1}\bigg(\frac{\Omega Q}{H}\bigg)$ as the collective coordinate in the  $\Omega_{(1)}$- modified Hamiltonian $H_{(2)}$, that is,
\begin{eqnarray}
H_{(2)}=\frac{1}{2}\bigg(\bigg[\frac{1}{2}((pq)^2+\Omega^2 \log(q)^2)-a_{(1)}\bigg]^2+\Omega_{(1)}^2 \bigg[\frac{1}{\Omega}\tan^{-1}\bigg(\frac{\Omega \log(q)}{pq}\bigg)\bigg]^2\bigg).\label {lam106a}
\end{eqnarray}
Now we can obtain the solutions for $p$ and $q$ as
\begin{eqnarray}
&&p(t)=\sqrt{2 a_{(1)}+2A_{(1)}\cos[\Omega_{(1)} t+\delta_{(1)}]}
\cos[\frac{\Omega A_{(1)}}{\Omega_{(1)}}\sin(\Omega_{(1)} t+\delta_{(1)})]/q(t),
\nonumber\\
&& q(t)=e^{\frac{1}{\Omega}\sqrt{2 a_{(1)}+2A_{(1)}\cos[\Omega_{(1)} t+\delta_{(1)}]}\sin[\frac{\Omega A_{(1)}}{\Omega_{(1)}}\sin(\Omega_{(1)} t+\delta_{(1)})]}.\label {lam106b}
\end{eqnarray} 
Choosing the arbitrary parameters such that $a_{1}>|A_{1}|$, the system for $p$ and $q$ is isochronous since the solution (\ref{lam106b}) is periodic with period $T=2\pi/\Omega_{(1)}$. This is in conformity with Theorem 2. 

Then we may extended the above analysis to the $\Omega_{(2)}$ modified Hamiltonian.  In this case   $H_{(2)}-a_{(2)}$ can be taken as the  momentum and $Q_{(2)}=\frac{1}{\Omega_{(1)}}\tan^{-1}\bigg(\frac{\Omega_{(1)} Q_{(1)}}{H_{(1)}-a_{(1)}}\bigg)$ as the conjugate coordinate and  therefore   
\begin{eqnarray}
H_{(3)}=& \,\frac{1}{2}\bigg[\frac{\Omega_{(2)}^2}{\Omega_{(1)}^2} \tan^{-1}[\frac{2\Omega_{(1)}\tan^{-1}[\frac{\Omega \log(q)}{pq}]}{\Omega((pq)^2+\Omega^2\log(q)^2-2a_{(1)})}]^2 +\frac{1}{4}\bigg(\frac{\Omega_{(1)}^2}{\Omega^2 }\tan^{-1}[\frac{\Omega \log(q)}{pq}]^2 \nonumber\\ &\qquad\qquad -2 a_{(2)}
+\frac{1}{4}  ((pq)^2+\Omega^2 \log(q)^2-2a_{(1)})^2 
\bigg)^2 \bigg].
\end{eqnarray}
 The solutions for $p$ and $q$ can now be written as
\begin{eqnarray}
p(t)= \bigg(\sqrt{2} \sqrt{ a_{(1)} + f_{(1)}\cos (f_{(2)})} 
  \cos (f_{(3)})\bigg )/q(t),\nonumber
\end{eqnarray}
\begin{eqnarray}
 q(t)=e^{\frac{1}{\Omega}\sqrt{2} \sqrt{ a_{(1)} + f_{(1)}\cos (f_{(2)}) } 
 \sin (f_{(3)})},
\end{eqnarray}
where $f_{(1)}= \bigg(2( a_{(2)}+A_{(2)} \cos[\Omega_{(2)} t + \delta_{(2)}] )\bigg)^{1/2}$ , $f_{(2)}= A_{(2)}\Omega_{(1)} \sin[\Omega_{(2)}t+\delta_{(2)}]/(\Omega_{(2)})$ and $f_{(3)}= (\Omega/\Omega_{(1)})f_{(1)}\sin[f_{(2)}]$.  We assume here again that $a_{2}>|A_{2}|$ and $a_{1}>\sqrt{2(a_{2}+A_{2})}$ so that $p$ and $q$ are real.  Here also the canonical variables $p$ and $q$ are periodic with period $T=2\pi/\Omega_{(2)}$ confirming the isochronous character of the dynamics. Following a similar analysis, one can generate more and more isochronous Hamiltonians.

To conclude, we have proposed a simple procedure to identify the collective coordinate $Q$ which is conjugate to the given Hamiltonian $H$ in order to generate isochronous systems. Using the known Hamiltonian $H$ and collective coordinate $Q$, we have proved the possibility of generating more and more isochronous oscillator systems recursively.

%\section*{Acknowledgments}
The work is supported by a Department of Science and Technology (DST), Government of India, Ramanna Fellowship program and a DST--IRHPA research project, Government of
India.

\end{document}